\documentclass{amsart}%
\usepackage{amsfonts}
\usepackage{amsmath}
\usepackage{amssymb}
\usepackage{graphicx}
\usepackage[latin1]{inputenc}
\usepackage[english]{babel}%
\newtheorem{theorem}{Theorem}[section]
\theoremstyle{plain}

\newtheorem{definition}[theorem]{Definition}

\newtheorem{lemma}[theorem]{Lemma}

\newtheorem{proposition}[theorem]{Proposition}
\newtheorem{remark}[theorem]{Remark}

\numberwithin{equation}{section}
\begin{document}
\title[Parabolic-type Equations and Ultrametric Random Walks]{Nonlocal Operators, Parabolic-type Equations, and Ultrametric Random Walks}
\author{L. F. Chacón-Cortes}
\author{W. A. Zúñiga-Galindo}
\address{Centro de Investigacion y de Estudios Avanzados del I.P.N., Departamento de
Matematicas, Av. Instituto Politecnico Nacional 2508, Col. San Pedro
Zacatenco, Mexico D.F., C.P. 07360, Mexico}
\email{fchaconc@math.cinvestav.edu.mx, wazuniga@math.cinvestav.edu.mx}
\thanks{The second author was partially supported by Conacyt (Mexico), Grant \# 127794.}

\begin{abstract}
In this article we introduce a new type of nonlocal operators and study the
Cauchy problem for certain parabolic-type pseudodifferential equations
naturally associated to these operators. Some of these equations are the
$p$-adic master equations of certain models of complex systems introduced by
Avetisov et al. The fundamental solutions of these parabolic-type equations
are transition functions of random walks on the $n$-dimensional vector space
over the field of $p$-adic numbers. We study some properties of these random
walks, including the first passage time.

\end{abstract}
\maketitle

\section{Introduction}

During the last twenty-five years there have been a strong interest on random
walks on ultrametric spaces mainly due its connections with models of complex
systems, such as glasses and proteins. Random walks on ultrametric spaces are
very convenient for describing phenomena whose space of states display a
hierarchical structure, see e.g. \cite{Av-1}-\cite{Av-7}, \cite{Dra-Kh-K-V},
\cite{Ga-Zu}, \cite{Ka}, \cite{Koch}, \cite{K-Kos}, \cite{M-P-V},
\cite{R-T-V}, \cite{Va1}, \cite{V-V-Z}, \cite{Zu}, and the references therein.
In the middle of the eighties G. Frauenfelder, G. Parisi, D. Stain, among
others, proposed using ultrametric spaces to describe the states of complex
biological systems, which possess a natural hierarchical organization, see
e.g. \cite{M-P-V}, \cite{R-T-V}. Avetisov et al. have constructed a wide
variety of models of ultrametric diffusion constrained by hierarchical energy
landscapes, see \cite{Av-1}-\cite{Av-7}. These models can be applied, among
other things, to the study the relaxation of biological complex systems
\cite{Av-4}. From a mathematical point view, in these models the
time-evolution of a complex system is described by a $p$-adic master equation
(a parabolic-type pseudodifferential equation)\ which controls the
time-evolution of a transition function of a random walk on an ultrametric
space, and the random walk describes the dynamics of the system in the space
of configurational states which is approximated by an ultrametric space
($\mathbb{Q}_{p}$).

This article continues and extends some of the mathematical results given in
\cite{Av-4}, \cite{Av-2}. We introduced a new class of nonlocal operators
which includes the Vladimirov operator. These operators are determined by a
radial function which determines the structure of the energy landscape being
used. In \cite{Av-4}\ several solvable models for degenerated landscapes of
types linear, logarithmic and exponential were studied. In this article we
\ study a large class of solvable models, we have called them polynomial and
exponential landscapes, which includes the linear and exponential landscapes
considered in \cite{Av-4}, see Section \ref{Complex_systems}. We attach to
each of these operators a Markov process (random walk) which is bounded and
has do discontinuities other than jumps, see Theorems \ref{Thm1}-\ref{Thm2}.
We also solve the Cauchy problem for the master equations attached to these
operators, see Theorem \ref{Thm3}. Finally, we study the first passage time
problem for the random walks attached to polynomial landscapes, see Theorem
\ref{Thm4}. All the results are formulated in arbitrary dimension.

The article is organized as follows. In section \ref{Section1} we review the
basic notions of $p$-adic analysis. In Section \ref{Sect2} we introduce a new
type of nonlocal operators, these operators encode the underlying energy
landscapes studied here. We show that these operators are pseudodifferential
and give some properties of their symbols. In Sections \ref{Sect3}-\ref{Sect4}
we study the Markov processes attached to the operators introduced in Section
\ref{Sect2}. In Section \ref{Sect5} we study the Cauchy problem for the master
equations which are pseudodifferential equations of parabolic-type. Finally,
in Section \ref{Sect6} we consider the problem of the first passage time for
random walks on $p$-adic spaces. A particular case of this problem was
considered earlier by Avetisov, Bikulov and Zubarev in \cite{Av-2}. In the
present article more general evolutions are considered and the regimes of
recurrent and transient random walks are investigated.

\section{\label{Section1}Preliminaries}

In this section we fix the notation and collect some basic results on $p$-adic
analysis that we will use through the article. For a detailed exposition on
$p$-adic analysis the reader may consult \cite{A-K-S}, \cite{Taibleson},
\cite{V-V-Z}.

\subsection{The field of $p$-adic numbers}

Along this article $p$ will denote a prime number different from $2$. The
field of $p-$adic numbers $\mathbb{Q}_{p}$ is defined as the completion of the
field of rational numbers $\mathbb{Q}$ with respect to the $p-$adic norm
$|\cdot|_{p}$, which is defined as
\[
|x|_{p}=%
\begin{cases}
0 & \text{if }x=0\\
p^{-\gamma} & \text{if }x=p^{\gamma}\dfrac{a}{b},
\end{cases}
\]
where $a$ and $b$ are integers coprime with $p$. The integer $\gamma:=ord(x)$,
with $ord(0):=+\infty$, is called the\textit{ }$p-$\textit{adic order of} $x$.
We extend the $p-$adic norm to $\mathbb{Q}_{p}^{n}$ by taking%
\[
||x||_{p}:=\max_{1\leq i\leq n}|x_{i}|_{p},\qquad\text{for }x=(x_{1}%
,\dots,x_{n})\in\mathbb{Q}_{p}^{n}.
\]
We define $ord(x)=\min_{1\leq i\leq n}\{ord(x_{i})\}$, then $||x||_{p}%
=p^{-\text{ord}(x)}$. Any $p-$adic number $x\neq0$ has a unique expansion
$x=p^{ord(x)}\sum_{j=0}^{\infty}x_{j}p^{j}$, where $x_{j}\in\{0,1,2,\dots
,p-1\}$ and $x_{0}\neq0$. By using this expansion, we define \textit{the
fractional part of }$x\in\mathbb{Q}_{p}$, denoted $\{x\}_{p}$, as the rational
number
\[
\{x\}_{p}=%
\begin{cases}
0 & \text{if }x=0\text{ or }ord(x)\geq0\\
p^{\text{ord}(x)}\sum_{j=0}^{-ord(x)-1}x_{j}p^{j} & \text{if }ord(x)<0.
\end{cases}
\]
For $\gamma\in\mathbb{Z}$, denote by $B_{\gamma}^{n}(a)=\{x\in\mathbb{Q}%
_{p}^{n}:||x-a||_{p}\leq p^{\gamma}\}$ \textit{the ball of radius }$p^{\gamma
}$ \textit{with center at} $a=(a_{1},\dots,a_{n})\in\mathbb{Q}_{p}^{n}$, and
take $B_{\gamma}^{n}(0):=B_{\gamma}^{n}$. Note that $B_{\gamma}^{n}%
(a)=B_{\gamma}(a_{1})\times\cdots\times B_{\gamma}(a_{n})$, where $B_{\gamma
}(a_{i}):=\{x\in\mathbb{Q}_{p}:|x-a_{i}|_{p}\leq p^{\gamma}\}$ is the
one-dimensional ball of radius $p^{\gamma}$ with center at $a_{i}\in
\mathbb{Q}_{p}$. The ball $B_{0}^{n}(0)$ is equals the product of $n$ copies
of $B_{0}(0):=\mathbb{Z}_{p}$, \textit{the ring of }$p-$\textit{adic
integers}. We denote by $\Omega(\left\Vert x\right\Vert _{p})$ the
characteristic function of $B_{0}^{n}(0)$. For more general sets, say Borel
sets, we use ${\LARGE 1}_{A}\left(  x\right)  $ to denote the characteristic
function of $A$.

\subsection{The Bruhat-Schwartz space}

A complex-valued function $\varphi$ defined on $\mathbb{Q}_{p}^{n}$ is
\textit{called locally constant} if for any $x\in\mathbb{Q}_{p}^{n}$ there
exist an integer $l(x)\in\mathbb{Z}$ such that%
\begin{equation}
\varphi(x+x^{\prime})=\varphi(x)\text{ for }x^{\prime}\in B_{l(x)}^{n}.
\label{local_constancy}%
\end{equation}
A function $\varphi:\mathbb{Q}_{p}^{n}\rightarrow\mathbb{C}$ is called a
\textit{Bruhat-Schwartz function (or a test function)} if it is locally
constant with compact support. The $\mathbb{C}$-vector space of
Bruhat-Schwartz functions is denoted by $S(\mathbb{Q}_{p}^{n}):=S$. For
$\varphi\in S(\mathbb{Q}_{p}^{n})$, the largest of such number $l=l(\varphi)$
satisfying (\ref{local_constancy}) is called \textit{the exponent of local
constancy of} $\varphi$.

Let $S^{\prime}(\mathbb{Q}_{p}^{n}):=S^{\prime}$ denote the set of all
functionals (distributions) on $S(\mathbb{Q}_{p}^{n})$. All functionals on
$S(\mathbb{Q}_{p}^{n})$ are continuous.

Set $\Psi(y)=\exp(2\pi i\{y\}_{p})$ for $y\in\mathbb{Q}_{p}$. The map
$\Psi(\cdot)$ is an additive character on $\mathbb{Q}_{p}$, i.e. a continuos
map from $\mathbb{Q}_{p}$ into the unit circle satisfying $\Psi(y_{0}%
+y_{1})=\Psi(y_{0})\Psi(y_{1})$, $y_{0},y_{1}\in\mathbb{Q}_{p}$.

Given $\xi=(\xi_{1},\dots,\xi_{n})$ and $x=(x_{1},\dots,x_{n})\in
\mathbb{Q}_{p}^{n}$, we set $\xi\cdot x:=\sum_{j=1}^{n}\xi_{j}x_{j}$. The
Fourier transform of $\varphi\in S(\mathbb{Q}_{p}^{n})$ is defined as
\[
(\mathcal{F}\varphi)(\xi)=\int_{\mathbb{Q}_{p}^{n}}\Psi(-\xi\cdot
x)\varphi(\xi)d^{n}x\quad\text{for }\xi\in\mathbb{Q}_{p}^{n},
\]
where $d^{n}x$ is the Haar measure on $\mathbb{Q}_{p}^{n}$ normalized by the
condition $vol(B_{0}^{n})=1$. The Fourier transform is a linear isomorphism
from $S(\mathbb{Q}_{p}^{n})$ onto itself satisfying $(\mathcal{F}%
(\mathcal{F}\varphi))(\xi)=\varphi(-\xi)$. We will also use the notation
$\mathcal{F}_{x\rightarrow\xi}\varphi$ and $\widehat{\varphi}$\ for the
Fourier transform of $\varphi$.

\subsubsection{Fourier transform}

The Fourier transform $\mathcal{F}\left[  f\right]  $ of a distribution $f\in
S^{\prime}\left(  \mathbb{Q}_{p}^{n}\right)  $ is defined by%
\[
\left(  \mathcal{F}\left[  f\right]  ,\varphi\right)  =\left(  f,\mathcal{F}%
\left[  \varphi\right]  \right)  \text{ for all }\varphi\in S\left(
\mathbb{Q}_{p}^{n}\right)  \text{.}%
\]
The Fourier transform $f\rightarrow\mathcal{F}\left[  f\right]  $ is a linear
isomorphism from $S^{\prime}\left(  \mathbb{Q}_{p}^{n}\right)  $\ onto
$S^{\prime}\left(  \mathbb{Q}_{p}^{n}\right)  $. Furthermore, $f=\mathcal{F}%
\left[  \mathcal{F}\left[  f\right]  \left(  -\xi\right)  \right]  $.

\section{\label{Sect2}A New Class of Nonlocal Operators}

Take \ $\mathbb{R}_{+}:=\left\{  x\in\mathbb{R};x\geq0\right\}  $, and fix a
function%
\[
w:\mathbb{Q}_{p}^{n}\rightarrow\mathbb{R}_{+}%
\]
satisfying the following properties:

\noindent(i) $w\left(  y\right)  $ is a radial (i.e. $w\left(  y\right)
=w\left(  \left\Vert y\right\Vert _{p}\right)  $)\ and continuous function;

\noindent(ii) $w\left(  y\right)  =0$ if and only if $y=0$;

\noindent(iii) there exists constants $C_{0}>0$, $M\in\mathbb{Z}$, and
$\alpha_{1}>n$ such that
\[
C_{0}\left\Vert y\right\Vert _{p}^{\alpha_{1}}\leq w(\left\Vert y\right\Vert
_{p})\text{, for }\left\Vert y\right\Vert _{p}\geq p^{M}.
\]

Note that condition (iii) implies that
\begin{equation}
{\int\limits_{\left\Vert y\right\Vert _{p}\geq p^{M}}}\frac{d^{n}y}{w\left(
\left\Vert y\right\Vert _{p}\right)  }<\infty. \label{Ec1}%
\end{equation}
In addition, since $w\left(  y\right)  $\ is a continuous function,
(\ref{Ec1}) holds for any $M\in\mathbb{Z}$. Convergence conditions for
integral kernels of type (\ref{Ec1}) were considered in \cite{Koz1},
\cite{Koz-Kh} and \cite{Kh-Koz2}.

We define%

\begin{equation}
(\boldsymbol{W}\varphi)(x)=\kappa{\int\limits_{\mathbb{Q}_{p}^{n}}}%
\frac{\varphi\left(  x-y\right)  -\varphi\left(  x\right)  }{w\left(
y\right)  }d^{n}y\text{, for }\varphi\in S\text{,} \label{W}%
\end{equation}
where $\kappa$ is a positive constant.

\begin{lemma}
\label{lemma1}For $1\leq\rho\leq\infty$,%
\[%
\begin{array}
[c]{ccc}%
S\left(  \mathbb{Q}_{p}^{n}\right)  & \rightarrow & L^{\rho}\left(
\mathbb{Q}_{p}^{n}\right) \\
&  & \\
\varphi & \rightarrow & \boldsymbol{W}\varphi
\end{array}
\]
is a well-defined linear operator. Furthermore,
\begin{equation}
\mathcal{F}\left[  \boldsymbol{W}\varphi\right]  (\xi)=-\kappa\left(
{\int\limits_{\mathbb{Q}_{p}^{n}}}\frac{1-\Psi\left(  -y\cdot\xi\right)
}{w\left(  y\right)  }d^{n}y\right)  \mathcal{F}\left[  \varphi\right]
\left(  \xi\right)  . \label{W_for}%
\end{equation}

\end{lemma}

\begin{proof}
Note that
\begin{equation}
(\boldsymbol{W}\varphi)(x)=\kappa\frac{{\LARGE 1}_{\mathbb{Q}_{p}%
^{n}\smallsetminus B_{p^{M}}^{n}}\left(  x\right)  }{w\left(  x\right)  }%
\ast\varphi\left(  x\right)  -\kappa\varphi\left(  x\right)  \left(  \text{
}{\int\limits_{\left\Vert y\right\Vert _{p}\geq p^{M}}}\frac{d^{n}y}{w\left(
y\right)  }\right)  , \label{Ec2}%
\end{equation}
for some constant $M=M(\varphi)$. If $\varphi\in S\subset L^{\rho}$, for
$1\leq\rho\leq\infty$, (\ref{Ec1}), then the Young inequality implies that the
first term on the right-hand side of (\ref{Ec2}) belongs to $L^{\rho}$\ for
$1\leq\rho\leq\infty$, and by (\ref{Ec1}) the second term in (\ref{Ec2}) also
belongs to $L^{\rho}$\ for $1\leq\rho\leq\infty$. Finally, formula
(\ref{W_for}) follows from Fubini's theorem, since
\[
\left\vert \frac{\varphi\left(  x-y\right)  -\varphi\left(  x\right)
}{w\left(  y\right)  }\right\vert \in L^{1}\left(  \mathbb{Q}_{p}^{n}%
\times\mathbb{Q}_{p}^{n},d^{n}xd^{n}y\right)  .
\]

\end{proof}

We set
\[
A_{w}\left(  \xi\right)  :={\int\limits_{\mathbb{Q}_{p}^{n}}}\frac
{1-\Psi\left(  -y\cdot\xi\right)  }{w\left(  y\right)  }d^{n}y.
\]

\begin{lemma}
\label{lemma2}The function $A_{w}\left(  \xi\right)  $ has the following
properties: (i) for $\left\Vert \xi\right\Vert _{p}=p^{-\gamma}\neq0$, with
$\gamma=ord(\xi)$,
\begin{equation}
A_{w}\left(  p^{-\gamma}\right)  =(1-p^{-n})\sum_{j=\gamma+2}^{\infty}%
\frac{p^{nj}}{w(p^{j})}+\frac{p^{n\gamma}}{w(p^{\gamma+1})}; \label{fomula_A}%
\end{equation}
(ii) it is radial, positive, continuous, and $A_{w}\left(  0\right)  =0$,
(iii) $A_{w}\left(  p^{-ord(\xi)}\right)  $ is a decreasing function of
$ord(\xi)$.
\end{lemma}

\begin{proof}
We write $\xi=p^{\gamma}\xi_{0},$ with $\gamma=ord(\xi)$ and $\left\Vert
\xi_{0}\right\Vert _{p}=1$. Then%

\begin{equation}
A_{w}\left(  \xi\right)  ={\int\limits_{\mathbb{Q}_{p}^{n}}}\frac
{1-\Psi\left(  -p^{\gamma}y\cdot\xi_{0}\right)  }{w\left(  \left\Vert
y\right\Vert _{p}\right)  }d^{n}y={p^{\gamma n}\int\limits_{\mathbb{Q}_{p}%
^{n}}}\frac{1-\Psi\left(  -z\cdot\xi_{0}\right)  }{w\left(  p^{\gamma
}\left\Vert z\right\Vert _{p}\right)  }d^{n}z. \label{Ec3}%
\end{equation}

We now note that
\[
\mathbb{Q}_{p}^{n}\smallsetminus\left\{  0\right\}  ={\bigsqcup\limits_{j\in
\mathbb{Z}}}p^{j}U
\]
with
\[
U:=\left\{  y\in\mathbb{Q}_{p}^{n};\left\Vert y\right\Vert _{p}=1\right\}  .
\]
By using this partition and (\ref{Ec3}), we have
\begin{align*}
A_{w}\left(  \xi\right)   &  ={\sum\limits_{j\in\mathbb{Z}}}p^{\gamma n}\text{
}{\int\limits_{p^{j}U}}\frac{1-\Psi\left(  -z\cdot\xi_{0}\right)  }{w\left(
p^{\gamma}\left\Vert z\right\Vert _{p}\right)  }d^{n}z\\
&  {=\sum\limits_{j\in\mathbb{Z}}}\text{ }\frac{p^{-jn+\gamma n}}{w\left(
p^{-j+\gamma}\right)  }\left\{  {\left(  1-p^{-n}\right)  -\int\limits_{U}%
}\Psi\left(  -p^{j}y\cdot\xi_{0}\right)  d^{n}y\right\}  .
\end{align*}
By using the formula
\begin{equation}
{\int\limits_{U}}\Psi\left(  -p^{j}y\cdot\xi_{0}\right)  d^{n}y=\left\{
\begin{array}
[c]{lll}%
1-p^{-n} & \text{if} & j\geq0\\
&  & \\
-p^{-n} & \text{if} & j=-1\\
&  & \\
0 & \text{if} & j<-1,
\end{array}
\right.  \label{formula}%
\end{equation}
we get
\begin{align}
A_{w}\left(  \xi\right)   &  =(1-p^{-n})\sum_{j=2}^{\infty}\frac
{p^{n(\gamma+j)}}{w(p^{\gamma+j})}+\frac{p^{n\gamma}}{w(p^{\gamma+1}%
)}\nonumber\\
&  =(1-p^{-n})\sum_{j=\gamma+2}^{\infty}\frac{p^{nj}}{w(p^{j})}+\frac
{p^{n\gamma}}{w(p^{\gamma+1})}. \label{Ec4}%
\end{align}
From (\ref{Ec4}) follows that $A_{w}\left(  \xi\right)  $ is radial, positive,
continuous outside of the origin, and that $A_{w}\left(  p^{-ord(\xi)}\right)
$ is a decreasing function of $ord(\xi)$. To show the continuity at origin, we
proceed as follows. Since $\sum_{j=M}^{\infty}\frac{p^{nj}}{w(p^{j})}<\infty$,
c.f. (\ref{Ec1}),
\[
A_{w}\left(  0\right)  :=\underset{\gamma\rightarrow\infty}{\lim}%
(1-p^{-n})\sum_{j=\gamma+2}^{\infty}\frac{p^{nj}}{w(p^{j})}+\underset
{\gamma\rightarrow\infty}{\lim}\frac{p^{n\gamma}}{w(p^{\gamma+1})}=0.
\]

\end{proof}

\begin{proposition}
\label{prop1}(i) $\left(  \boldsymbol{W}\varphi\right)  \left(  x\right)
=-\kappa\mathcal{F}_{\xi\rightarrow x}^{-1}\left(  A_{w}(\left\Vert
\xi\right\Vert _{p})\mathcal{F}_{x\rightarrow\xi}\varphi\right)  $ for
$\varphi\in S\left(  \mathbb{Q}_{p}^{n}\right)  $, and $\boldsymbol{W}%
\varphi\in C\left(  \mathbb{Q}_{p}^{n}\right)  \cap L^{\rho}\left(
\mathbb{Q}_{p}^{n}\right)  $, for $1\leq\rho\leq\infty$. The Operator
$\boldsymbol{W}$ extends to an unbounded and densely defined operator in
$L^{2}\left(  \mathbb{Q}_{p}^{n}\right)  $ with domain
\begin{equation}
Dom(\boldsymbol{W})=\left\{  \varphi\in L^{2};A_{w}(\left\Vert \xi\right\Vert
_{p})\mathcal{F}\varphi\in L^{2}\right\}  . \label{Dom_W}%
\end{equation}

\noindent(ii) $\left(  -\boldsymbol{W},Dom(\boldsymbol{W})\right)  $ is
self-adjoint and positive operator.

\noindent(iii) $\boldsymbol{W}$ is the infinitesimal generator of a
contraction $C_{0}$ semigroup $\left(  \mathcal{T}(t)\right)  _{t\geq0}$.
Moreover, the semigroup $\left(  \mathcal{T}(t)\right)  _{t\geq0}$ is bounded
holomorphic with angle $\pi/2$.
\end{proposition}

\begin{proof}
(i) It follows from Lemma \ref{lemma1} and the fact that $A_{w}(\left\Vert
\xi\right\Vert _{p})$ is continuous, c.f. Lemma \ref{lemma2}. (ii) follows
from the fact that $\boldsymbol{W}$ is a pseudodifferential operator and that
the Fourier transform preserves the inner product of $L^{2}$. (iii) It follows
of well-known results, see e.g. \cite[Chap. 2, Sect. 3]{E-N} or \cite{C-H}.
For the property of the semigroup of being holomorphic, see e.g. \cite[Chap.
2, Sect. 4.7]{E-N}.
\end{proof}

\subsection{Some additional results}

\begin{lemma}
\label{lemma3}Assume that there exist positive constants $\alpha_{1}%
$,$\,\alpha_{2}$,$\ C_{0}$, $C_{1}$, with $\alpha_{1}>n$, $\alpha_{2}>n$, and
$\alpha_{3}\geq0$, such that
\begin{equation}
C_{0}\left\Vert \xi^{\prime}\right\Vert _{p}^{\alpha_{1}}\leq w(\left\Vert
\xi^{\prime}\right\Vert _{p})\leq C_{1}\left\Vert \xi^{\prime}\right\Vert
_{p}^{\alpha_{2}}e^{\alpha_{3}\left\Vert \xi^{\prime}\right\Vert _{p}}\text{,
for any }\xi^{\prime}\in\mathbb{Q}_{p}^{n}. \label{Ec5}%
\end{equation}
Then there exist positive constants $C_{2}$, $C_{3}$, such that
\[
C_{2}\left\Vert \xi\right\Vert _{p}^{\alpha_{2}-n}e^{-\alpha_{3}p\left\Vert
\xi\right\Vert _{p}^{-1}}\leq A_{w}(\left\Vert \xi\right\Vert _{p})\leq
C_{3}\left\Vert \xi\right\Vert _{p}^{\alpha_{1}-n}%
\]
for any $\xi\in\mathbb{Q}_{p}^{n}$, with the convention that $e^{-\alpha
_{3}p\left\Vert 0\right\Vert _{p}^{-1}}:=\lim_{\left\Vert \xi\right\Vert
_{p}\rightarrow0}e^{-\alpha_{3}p\left\Vert \xi\right\Vert _{p}^{-1}}=0$.
Furthermore, if $\alpha_{3}>0$, then $\alpha_{1}\geq\alpha_{2}$, and if
$\alpha_{3}=0$, then $\alpha_{1}=\alpha_{2}$.
\end{lemma}

\begin{proof}
By using the lower bound for $w$ given in (\ref{Ec5}), and $\left\Vert
\xi\right\Vert _{p}=p^{-\gamma}$,%

\[
A_{w}(\left\Vert \xi\right\Vert _{p})\leq\frac{(1-p^{-n})}{C_{0}}%
\sum_{j=\gamma+2}^{\infty}\frac{p^{nj}}{p^{j\alpha_{1}}}+\frac{p^{n\gamma}%
}{p^{\alpha_{1}(\gamma+1)}}\leq C_{3}\left\Vert \xi\right\Vert _{p}%
^{\alpha_{1}-n}.
\]

On the other hand, $A_{w}\left(  \left\Vert \xi\right\Vert _{p}\right)
\geq\frac{p^{n\gamma}}{w(p^{\gamma+1})}$ , and by using the upper bound for
$w$ given in (\ref{Ec5}),%
\[
A_{w}\left(  \left\Vert \xi\right\Vert _{p}\right)  \geq\frac{p^{n\gamma}%
}{w(p^{\gamma+1})}\geq\frac{p^{n\gamma}}{C_{1}p^{\alpha_{2}(\gamma
+1)}e^{\alpha_{3}p^{\gamma+1}}}\geq C_{2}\left\Vert \xi\right\Vert
_{p}^{\alpha_{2}-n}e^{-\alpha_{3}p\left\Vert \xi\right\Vert _{p}^{-1}}.
\]

\end{proof}

\begin{definition}
\label{Def_exp_pol_W}We \ say that $\boldsymbol{W}$ (or $A_{w}$) is of
exponential type if inequality (\ref{Ec5}) is only possible for $\alpha_{3}>0$
with $\alpha_{1}$,$\,\alpha_{2}$,$\ C_{0}$, $C_{1}$ positive constants and
$\alpha_{1}>n$, $\alpha_{2}>n$. If (\ref{Ec5}) holds for $\alpha_{3}=0$ with
$\alpha_{1}$,$\,\alpha_{2}$,$\ C_{0}$, $C_{1}$ positive constants and
$\alpha_{1}>n$, $\alpha_{2}>n$, we say that $\boldsymbol{W}$ (or $A_{w}$) is
of polynomial type.
\end{definition}

We note that if $\boldsymbol{W}$ is of polynomial type then $\alpha
_{1}=\,\alpha_{2}>n\ $and $C_{0}$, $C_{1}$ are positive constants with
$C_{1}\geq C_{0}$.

\begin{lemma}
\label{lemma4}With the hypotheses of Lemma \ref{lemma3},
\[
e^{-t\kappa A_{w}(\left\Vert \xi\right\Vert _{p})}\in L^{\rho}(\mathbb{Q}%
_{p}^{n})\text{\ for }1\leq\rho<\infty\text{ and }t>0.
\]

\end{lemma}

\begin{proof}
Since $e^{-tA_{w}(\left\Vert \xi\right\Vert _{p})}$ is a continuous function,
it is sufficient to show that there exists $M\in\mathbb{N}$ such that
\[
I_{M}\left(  t\right)  :=\underset{\left\Vert \xi\right\Vert _{p}>p^{M}}{\int
}e^{-\rho\kappa tA_{w}(\left\Vert \xi\right\Vert _{p})}d^{n}\xi<\infty,\text{
for }t>0\text{.}%
\]
Take $M\in\mathbb{N}$, by Lemma \ref{lemma3}, we have
\[
C_{2}\left\Vert \xi\right\Vert _{p}^{\alpha_{2}-n}e^{-\alpha_{3}p\left\Vert
\xi\right\Vert _{p}^{-1}}>C_{2}\left\Vert \xi\right\Vert _{p}^{\alpha_{2}%
-n}e^{-\alpha_{3}p^{-M+1}}\text{ for }\left\Vert \xi\right\Vert _{p}>p^{M},
\]
and (with $B=C_{2}\kappa e^{-\alpha_{3}p^{-M+1}}$),
\[
I_{M}\left(  t\right)  \leq\underset{\left\Vert \xi\right\Vert _{p}>p^{M}%
}{\int}e^{-tB\left\Vert \xi\right\Vert _{p}^{\alpha_{2}-n}}d^{n}\xi\leq
C(M,\kappa,\rho)t^{\frac{-n}{\alpha_{2}-n}},\text{ for }t>0\text{.}%
\]

\end{proof}

\subsection{\label{Complex_systems}$p$-adic description of characteristic
relaxation in complex systems}

In \cite{Av-4} Avetisov et al. developed a new approach to the description of
relaxation processes in complex systems (such as glasses, macromolecules and
proteins) on the basis of $p$-adic analysis. The dynamics of a complex system
is described by a random walk in the space of configurational states, which is
approximated by an ultrametric space ($\mathbb{Q}_{p}$). Mathematically
speaking, the time- evolution of the system is controlled by a master equation
of the form
\begin{equation}
\frac{\partial f\left(  x,t\right)  }{\partial t}=%
{\displaystyle\int\limits_{\mathbb{Q}_{p}}}
\left\{  v\left(  x\mid y\right)  f\left(  y,t\right)  -v\left(  y\mid
x\right)  f\left(  x,t\right)  \right\}  dy\text{, }x\in\mathbb{Q}_{p}\text{,
}t\in\mathbb{R}_{+}, \label{Master_E}%
\end{equation}
where the function $f\left(  x,t\right)  :\mathbb{Q}_{p}\times\mathbb{R}%
_{+}\rightarrow\mathbb{R}_{+}$ is a probability density distribution, and the
function $v\left(  x\mid y\right)  :\mathbb{Q}_{p}\times\mathbb{Q}%
_{p}\rightarrow\mathbb{R}_{+}$ is the probability of transition from state $y$
to the state $x$ per unit time. \ The transition from a state $y$ to a state
$x$ can be perceived as overcoming the energy barrier \ separating these
states. In \cite{Av-4} an Arrhenius type relation was used:%
\[
v\left(  x\mid y\right)  \sim A(T)\exp\left\{  -\frac{U\left(  x\mid y\right)
}{kT}\right\}  ,
\]
where $U\left(  x\mid y\right)  $ is the height of the activation barrier for
the transition from the state $y$ to state $x$, $k$ is the Boltzmann constant
and $T$ is the temperature. This formula establishes a relation between the
structure of \textit{the energy landscape} $U\left(  x\mid y\right)  $ and the
transition function $v\left(  x\mid y\right)  $. The case $v\left(  x\mid
y\right)  =v\left(  y\mid x\right)  $ corresponds to a \textit{degenerate
energy landscape}. In this case the master equation (\ref{Master_E}) takes the
form%
\[
\frac{\partial f\left(  x,t\right)  }{\partial t}=%
{\displaystyle\int\limits_{\mathbb{Q}_{p}}}
v\left(  \left\vert x-y\right\vert _{p}\right)  \left\{  f\left(  y,t\right)
-f\left(  x,t\right)  \right\}  dy\text{,}%
\]
where $v\left(  \left\vert x-y\right\vert _{p}\right)  =\frac{A(T)}{\left\vert
x-y\right\vert _{p}}\exp\left\{  -\frac{U\left(  \left\vert x-y\right\vert
_{p}\right)  }{kT}\right\}  $. By choosing $U$ conveniently, several energy
landscapes can be obtained. Following \cite{Av-4}, there are three basic
landscapes: (i) (logarithmic) $v\left(  \left\vert x-y\right\vert _{p}\right)
=\frac{1}{\left\vert x-y\right\vert _{p}\ln^{\alpha}\left(  1+\left\vert
x-y\right\vert _{p}\right)  }$, $\alpha>1$ (ii) (linear) $v\left(  \left\vert
x-y\right\vert _{p}\right)  =\frac{1}{\left\vert x-y\right\vert _{p}%
^{\alpha+1}}$, $\alpha>0$, (iii) (exponential) $v\left(  \left\vert
x-y\right\vert _{p}\right)  =\frac{e^{-\alpha\left\vert x-y\right\vert _{p}}%
}{\left\vert x-y\right\vert _{p}}$, $\alpha>0$.

Thus, it is natural to study the following Cauchy problem:%

\[
\left\{
\begin{array}
[c]{lll}%
\frac{\partial u\left(  x,t\right)  }{\partial t}=\kappa{\int
\limits_{\mathbb{Q}_{p}^{n}}}\frac{u\left(  x-y,t\right)  -u\left(
x,t\right)  }{w\left(  y\right)  }d^{n}y\text{, } & x\in\mathbb{Q}_{p}^{n}, &
t\in\mathbb{R}_{+},\\
&  & \\
u\left(  x,0\right)  =\varphi\in S\left(  \mathbb{Q}_{p}^{n}\right)  , &  &
\end{array}
\right.
\]
where $w\left(  y\right)  $ is a radial function belonging to a class of
functions that contains functions like:

\noindent(i) $w(\left\Vert y\right\Vert _{p})=\Gamma_{p}^{n}(-\alpha
)\left\Vert y\right\Vert _{p}^{\alpha+n}$, here $\Gamma_{p}^{n}(\cdot)$ is the
$n$-dimensional $p$-adic Gamma function, and $\alpha>0$;

\noindent(ii) $w(\left\Vert y\right\Vert _{p})=\left\Vert y\right\Vert
_{p}^{\beta}e^{\alpha\left\Vert y\right\Vert _{p}}$, $\alpha>0$.

We recall that operator $\boldsymbol{W}$ corresponding to case (i) is the
Taibleson operator which is a generalization of the Vladimirov operator, see
\cite{R-Zu}.

By imposing condition (\ref{Ec5}) to $w$, we include the basic energies
landscapes in our study. Take $w(\left\Vert y\right\Vert _{p})$ satisfying
(\ref{Ec5}) and \ take $f\left(  \left\Vert y\right\Vert _{p}\right)  $ a
continuous function such that%
\[
0<\sup_{y\in\mathbb{Q}_{p}^{n}}f\left(  \left\Vert y\right\Vert _{p}\right)
<\infty\text{ and }0<\inf_{y\in\mathbb{Q}_{p}^{n}}f\left(  \left\Vert
y\right\Vert _{p}\right)  <\infty.
\]
Then $f\left(  \left\Vert y\right\Vert _{p}\right)  w(\left\Vert y\right\Vert
_{p})$ satisfies (\ref{Ec5}).

On the other hand, take $P(\left\Vert y\right\Vert _{p})$ to be a polynomial
in $\left\Vert y\right\Vert _{p}$ with real positive coefficients and nonzero
constant term, thus $\inf_{y\in\mathbb{Q}_{p}^{n}}P\left(  \left\Vert
y\right\Vert _{p}\right)  =P(0)>0$, and take $w(\left\Vert y\right\Vert
_{p})=\left\Vert y\right\Vert _{p}^{\beta}e^{\alpha\left\Vert y\right\Vert
_{p}}$ satisfying (\ref{Ec5}), then $P\left(  \left\Vert y\right\Vert
_{p}\right)  w(\left\Vert y\right\Vert _{p})$ also satisfies (\ref{Ec5}).

Finally we note that $\left\Vert y\right\Vert _{p}^{\beta}\ln^{\alpha
}(1+\left\Vert y\right\Vert _{p})$, $\beta>n$, $\alpha\in\mathbb{N}$, does not
satisfies $\left\Vert y\right\Vert _{p}^{\alpha_{1}}\leq\left\Vert
y\right\Vert _{p}^{\beta}\ln^{\alpha}(1+\left\Vert y\right\Vert _{p})$ for any
$y\in\mathbb{Q}_{p}^{n}$, and hence our results do not include the case of
logarithmic landscapes.

\section{\label{Sect3}Heat Kernels}

In this section we assume that function $w$ satisfies conditions (\ref{Ec5}).

We define
\[
Z(x,t;w,\kappa):=Z(x,t)=\underset{\mathbb{Q}_{p}^{n}}{\int}e^{-\kappa
tA_{w}(\left\Vert \xi\right\Vert _{p})}\Psi(x\cdot\xi)d^{n}\xi\text{ for
}t>0\text{ and }x\in\mathbb{Q}_{p}^{n}\text{.}%
\]
Note that by Lemma \ref{lemma4}, $Z(x,t)=\mathcal{F}_{\xi\rightarrow x}%
^{-1}[e^{-\kappa tA_{w}(\left\Vert \xi\right\Vert _{p})}]\in L^{1}\cap L^{2}$
for $t>0$. We call a such function a \textit{heat kernel}. When considering
$Z(x,t)$ as a function

of $x$ for $t$ fixed we will write $Z_{t}(x)$.

\begin{lemma}
\label{lemma5}There exists a positive constant $C$, such that%
\[
Z(x,t)<Ct\left\Vert x\right\Vert _{p}^{-\alpha_{1}}\text{, for }x\in
\mathbb{Q}_{p}^{n}\smallsetminus\left\{  0\right\}  \text{ and }t>0\text{.}%
\]

\end{lemma}

\begin{proof}
Let $\left\Vert x\right\Vert _{p}=p^{\beta}$. Since $Z(x,t)\in L^{1}%
(\mathbb{Q}_{p}^{n})$ for $t>0$, by applying the formula for the Fourier
transform of radial function, we get
\[
Z(x,t)=\left\Vert x\right\Vert _{p}^{-n}\left[  (1-p^{-n})\overset{\infty
}{\sum_{j=0}}e^{-\kappa A_{w}(p^{-\beta-j})t}p^{-nj}-e^{-\kappa A_{w}%
(p^{-\beta+1})t}\right]  .
\]
By using that $e^{-\kappa A_{w}(p^{-\beta-j})t}\leq1$ for $j\in\mathbb{N}$, we
have
\[
Z(x,t)\leq\left\Vert x\right\Vert _{p}^{-n}\left[  1-e^{-\kappa A_{w}%
(p^{-\beta+1})t}\right]  .
\]

We now apply the Mean Value Theorem to the real function $f\left(  u\right)
=e^{-\kappa A_{w}(p^{-\beta+1})u}$ on $\left[  0,t\right]  $ with $t>0$, and
Lemma \ref{lemma3},
\[
Z(x,t)\leq C_{0}\left\Vert x\right\Vert _{p}^{-n}tA_{w}(p^{-\beta+1})\leq
Ct\left\Vert x\right\Vert _{p}^{-\alpha_{1}}.
\]

\end{proof}

\begin{lemma}
\label{lemma6}\textup{$Z(x,t)\geq0$, }for $x\in\mathbb{Q}_{p}^{n}$ and $t>0$.
\end{lemma}

\begin{proof}
Since $e^{-tA_{w}(\left\Vert \xi\right\Vert _{p})}$ is integrable for $t>0$
and radial, we have%

\begin{align*}
Z(x,t)  &  =\sum_{i=-\infty}^{\infty}e^{-tA_{w}(p^{i})}\underset{\left\Vert
\xi\right\Vert _{p}=p^{i}}{\int}\Psi(x\cdot\xi)d^{n}\xi\\
&  =\sum_{i=-\infty}^{\infty}p^{ni}\left[  e^{-\kappa tA_{w}(p^{i}%
)}-e^{-\kappa tA_{w}(p^{i+1})}\right]  \Omega(\left\Vert p^{-i}x\right\Vert
_{p})\geq0
\end{align*}
since $A_{w}$ is increasing function of $i$, c.f. Lemma \ref{lemma2}.
\end{proof}

\begin{theorem}
\label{Thm1}The function $Z(x,t)$ has the following properties:

\noindent(i) $Z(x,t)\geq0$ for any $t>0$;

\noindent(ii)$\underset{\mathbb{Q}_{p}^{n}}{\int}Z(x,t)d^{n}x=1$ for any $t>0$;

\noindent(iii) $Z_{t}(x)\in C(\mathbb{Q}_{p}^{n},\mathbb{R})\cap
L^{1}(\mathbb{Q}_{p}^{n})\cap L^{2}(\mathbb{Q}_{p}^{n})$ for any $t>0$;

\noindent(iv) $Z_{t}(x)\ast Z_{t^{\prime}}(x)=Z_{t+t^{\prime}}(x)$ for any
$t$, $t^{\prime}>0$;

\noindent(v) $\underset{t\rightarrow0^{+}}{\lim}Z(x,t)=\delta(x)$ in
$S^{\prime}(\mathbb{Q}_{p}^{n})$.
\end{theorem}

\begin{proof}
(i) It follows from Lemma \ref{lemma6}. (ii) For any $t>0$ the function
$e^{-\kappa tA_{w}(\left\Vert \xi\right\Vert _{p})}$ is continuous at $\xi=0$
and by Lemma \ref{lemma4} we have $e^{-\kappa tA_{w}(\left\Vert \xi\right\Vert
_{p})}\in L^{1}\cap L^{2}$ for $t>0$, then $Z_{t}(x)\in L^{1}\cap L^{2}$ for
$t>0$. Now the statement follows from the inversion formula for the Fourier
transform. (iii) From Lemma \ref{lemma4}, with $\rho=1,2$, we have
$Z_{t}(x)\in C(\mathbb{Q}_{p}^{n},\mathbb{R})\cap L^{1}(\mathbb{Q}_{p}^{n})$,
$t>0$, and by (i) and (ii), $Z_{t}(x)\in L^{2}(\mathbb{Q}_{p}^{n})$. (iv) By
the previous property $Z_{t}(x)\in L^{1}$ for any $t>0$, then%
\begin{align*}
Z_{t}(x)\ast Z_{t^{\prime}}(x)  &  =\mathcal{F}_{\xi\rightarrow x}^{-1}\left(
e^{-\kappa tA_{w}(\left\Vert \xi\right\Vert _{p})}e^{-\kappa t^{\prime}%
A_{w}(\left\Vert \xi\right\Vert _{p})}\right)  =\mathcal{F}_{\xi\rightarrow
x}^{-1}\left(  e^{-\kappa\left(  t+t^{\prime}\right)  A_{w}(\left\Vert
\xi\right\Vert _{p})}\right) \\
&  =Z_{t+t^{\prime}}(x).
\end{align*}

(v) Since we have $e^{-\kappa tA_{w}(\left\Vert \xi\right\Vert _{p})}\in
C(\mathbb{Q}_{p}^{n},\mathbb{R})\cap L^{1}$ for $t>0$, c.f. Lemma
\ref{lemma4}, the inner product%
\[
\left\langle e^{-\kappa tA_{w}(\left\Vert \xi\right\Vert _{p})},\phi
\right\rangle =%
{\displaystyle\int\limits_{\mathbb{Q}_{p}^{n}}}
e^{-\kappa tA_{w}(\left\Vert \xi\right\Vert _{p})}\overline{\phi\left(
\xi\right)  }d^{n}\xi
\]
defines a distribution on $\mathbb{Q}_{p}^{n}$, then, by the Dominated
Converge Theorem,
\[
\lim_{t\rightarrow0^{+}}\left\langle e^{-\kappa tA_{w}(\left\Vert
\xi\right\Vert _{p})},\phi\right\rangle =\left\langle 1,\phi\right\rangle
\]
and thus
\[
\lim_{t\rightarrow0^{+}}\left\langle Z\left(  x,t\right)  ,\phi\right\rangle
=\lim_{t\rightarrow0^{+}}\left\langle e^{-\kappa tA_{w}(\left\Vert
\xi\right\Vert _{p})},\mathcal{F}^{-1}\phi\right\rangle =\left\langle
1,\mathcal{F}^{-1}\phi\right\rangle =\left(  \delta,\phi\right)  .
\]

\end{proof}

\section{\label{Sect4}Markov Processes over $\mathbb{Q}_{p}^{n}$}

Along this section we consider $\left(  \mathbb{Q}_{p}^{n},\left\Vert
\cdot\right\Vert _{p}\right)  $ as complete non-Archimedean metric space and
use the terminology and results of \cite[Chapters Two, Three]{Dyn}. Let
$\mathcal{B}$ denote the Borel $\sigma-$algebra of $\mathbb{Q}_{p}^{n}$. Thus
$\left(  \mathbb{Q}_{p}^{n},\mathcal{B},d^{n}x\right)  $ is a measure space.

We set%
\[
p(t,x,y):=Z(x-y,t)\ \text{for }t>0\text{,}\;x,y\in\mathbb{Q}_{p}^{n},
\]
and
\[
P(t,x,B)=%
\begin{cases}
\int_{B}p(t,y,x)d^{n}y & \text{for }t>0,\quad x\in\mathbb{Q}_{p}^{n},\quad
B\in\mathcal{B}\\
\boldsymbol{1}_{B}(x) & \text{for }t=0.
\end{cases}
\]

\begin{lemma}
\label{lemma7}With the above notation the following assertions hold:

(i) $p(t,x,y)$ is a normal transition density;

(ii) $P(t,x,B)$ is a normal transition function.
\end{lemma}

\begin{proof}
The result follows from Theorem \ref{Thm1}, see \cite[Section 2.1]{Dyn} for
further details.
\end{proof}

\begin{lemma}
\label{lemma8}The transition function $P(t,x,B)$ satisfies the following two conditions:

\noindent(i) for each $u\geq0$ and compact $B$%
\[
\underset{x\rightarrow\infty}{\lim}\underset{t\leq u}{\text{ }\sup
}P(t,x,B)=0\ \text{[Condition L(B)];}%
\]

\noindent(ii) for each $\epsilon>0$ and compact $B$%
\[
\underset{t\rightarrow0^{+}}{\lim}\underset{x\in B}{\sup}P(t,x,\mathbb{Q}%
_{p}^{n}\setminus B_{\epsilon}^{n}(x))=0\ \text{[Condition M(B)].}%
\]

\end{lemma}

\begin{proof}
(i) By Lemma \ref{lemma5} and the fact that $\left\Vert \cdot\right\Vert _{p}$
is an ultranorm, we have
\[
P(t,x,B)\leq Ct\underset{B}{\int}\left\Vert x-y\right\Vert _{p}^{-\alpha_{1}%
}d^{n}y=tC\left\Vert x\right\Vert _{p}^{-\alpha_{1}}vol\left(  B\right)
\text{ for }x\in\mathbb{Q}_{p}^{n}\setminus B.
\]

Therefore $\underset{x\rightarrow\infty}{\lim\text{ }}\underset{t\leq u}{\sup
}P(t,x,B)=0$.

(ii) Again, by Lemma \ref{lemma5}, the fact that $\left\Vert \cdot\right\Vert
_{p}$ is an ultranorm, and $\alpha_{1}>n$, we have%

\begin{align*}
P(t,x,\mathbb{Q}_{p}^{n}\setminus B_{\epsilon}^{n}(x))  &  \leq Ct\underset
{\left\Vert x-y\right\Vert _{p}>\epsilon}{\int}\left\Vert x-y\right\Vert
_{p}^{-\alpha_{1}}d^{n}y=Ct\underset{\left\Vert z\right\Vert _{p}>\epsilon
}{\int}\left\Vert z\right\Vert _{p}^{-\alpha_{1}}d^{n}z\\
&  =C^{\prime}\left(  \alpha_{1},\epsilon,n\right)  t.
\end{align*}

Therefore
\[
\underset{t\rightarrow0^{+}}{\lim}\underset{x\in B}{\sup}P(t,x,\mathbb{Q}%
_{p}^{n}\setminus B_{\epsilon}^{n}(x))\leq\underset{t\rightarrow0^{+}}{\lim
}\underset{x\in B}{\sup}C^{\prime}\left(  \alpha_{1},\epsilon,n\right)  t=0.
\]

\end{proof}

\begin{theorem}
\label{Thm2}$Z(x,t)$ is the transition density of a time and space homogeneous
Markov process which is bounded, right-continuous and has no discontinuities
other than jumps.
\end{theorem}

\begin{proof}
The result follows from \cite[Theorem 3.6]{Dyn} by using that $(\mathbb{Q}%
_{p}^{n},\left\Vert x\right\Vert _{p})$ is semi-compact space, i.e. a locally
compact Hausdorff space with a countable base, and $P(t,x,B)$ is a normal
transition function satisfying conditions $L(B)$ and $M(B)$, c.f. Lemmas
\ref{lemma7}, \ref{lemma8}.
\end{proof}

\section{\label{Sect5}The Cauchy Problem}

Consider the following Cauchy problem:%
\begin{equation}
\left\{
\begin{array}
[c]{ll}%
\frac{\partial u}{\partial t}(x,t)-\boldsymbol{W}u(x,t)=0, & x\in
\mathbb{Q}_{p}^{n},t\in\left[  0,\infty\right)  ,\\
& \\
u\left(  x,0\right)  =u_{0}(x), & u_{0}(x)\in Dom(\boldsymbol{W}),
\end{array}
\right.  \label{Cauchy1}%
\end{equation}
where $\left(  \boldsymbol{W}\phi\right)  \left(  x\right)  =-\kappa
\mathcal{F}_{\xi\rightarrow x}^{-1}\left(  A_{w}\left(  \left\Vert
w\right\Vert _{p}\right)  \mathcal{F}_{x\rightarrow\xi}\phi\right)  $ for
$\phi\in Dom(\boldsymbol{W})$, see (\ref{Dom_W}), \ and $u:$ $\mathbb{Q}%
_{p}^{n}\times\left[  0,\infty\right)  \rightarrow\mathbb{C}$ is an unknown
function. We say that a function $u(x,t)$ is a\textit{ solution} of
(\ref{Cauchy1}), if $u(x,t)\in C\left(  \left[  0,\infty\right)
,Dom(\boldsymbol{W})\right)  \cap C^{1}\left(  \left[  0,\infty\right)
,L^{2}(\mathbb{Q}_{p}^{n})\right)  $ and $u$ satisfies (\ref{Cauchy1}) for all
$t\geq0$.

In this section, we understand the notions of continuity in $t$,
differentiability in $t$ and equalities in the \ $L^{2}(\mathbb{Q}_{p}^{n})$
sense, as it is customary in the semigroup theory.

We know from Proposition \ref{prop1} that the operator $\boldsymbol{W}$
generates a $C_{0}$ semigroup $\left(  \mathcal{T}(t)\right)  _{t\geq0}$, then
Cauchy problem (\ref{Cauchy1}) is well-posed, i.e. it is uniquely solvable
with the solution continuously dependent on the initial data, and its solution
is given by $u(x,t)=\mathcal{T}(t)u_{0}(x)$, for $t\geq0$, see e.g.
\cite[Theorem 3.1.1]{C-H}. However the general theory does not give an
explicit formula for the semigroup $\left(  \mathcal{T}(t)\right)  _{t\geq0}$.
We show that the operator $\mathcal{T}(t)$ for $t>0$ coincides with the
operator of convolution with the heat kernel $Z_{t}\ast\cdot$. In order to
prove this, we first construct a solution of Cauchy problem (\ref{Cauchy1})
with the initial value from $S$ without using the semigroup theory. Then we
extend the result to all initial values from $Dom(\boldsymbol{W})$, see
Propositions \ref{prop2}-\ref{prop3}.

\subsection{Homogeneous equations with initial values in $S$}

To simplify the notation, set $Z_{0}\ast u_{0}=\left(  Z_{t}(x)\ast
u_{0}(x)\right)  \mid_{t=0}:=u_{0}$. We define the function
\begin{equation}
u\left(  x,t\right)  =Z_{t}(x)\ast u_{0}(x),\text{ for }t\geq0.
\label{funct_u}%
\end{equation}
Since $Z_{t}(x)\in L^{1}$ for $t>0$ and $u_{0}\in S(\mathbb{Q}_{p}^{n})\subset
L^{\infty}(\mathbb{Q}_{p}^{n})$, the convolution exists and is a continuous
function, see e.g. \cite[Theorem 1.1.6]{Rudin}.

\begin{lemma}
\label{lemma9}Take $u_{0}\in S$ with the support of $\ \widehat{u_{0}}$
contained in $B_{R}$, and $u\left(  x,t\right)  $, $t\geq0$ defined as in
(\ref{funct_u}). Then the following assertions hold:

\noindent(i) $u\left(  x,t\right)  $ is continuously differentiable in time
for $t\geq0$ and the derivative is given by%

\[
\frac{\partial u(x,t)}{\partial t}=-\kappa\mathcal{F}_{\xi\rightarrow x}%
^{-1}\left(  e^{-\kappa tA_{w}(\left\Vert \xi\right\Vert _{p})}A_{w}%
(\left\Vert \xi\right\Vert _{p})1_{B_{R}}(\xi)\right)  \ast u_{0}(x);
\]

\noindent(ii) $u(x,t)\in Dom(\boldsymbol{W})$ for any $t\geq0$ and%
\[
(\boldsymbol{W}u)(x,t)=-\kappa\mathcal{F}_{\xi\rightarrow x}^{-1}\left(
e^{-\kappa tA_{w}(\left\Vert \xi\right\Vert _{p})}A_{w}(\left\Vert
\xi\right\Vert _{p})1_{B_{R}}(\xi)\right)  \ast u_{0}(x).
\]

\end{lemma}

\begin{proof}
(i) The proof is similar to the one given for Lemma 7.1 in \cite{T-Z}. (ii)
Note that $e^{-\kappa tA_{w}(\left\Vert \xi\right\Vert _{p})}\widehat{u_{0}%
}\left(  \xi\right)  \in S$ for $t\geq0$, $A_{w}(\left\Vert \xi\right\Vert
_{p})e^{-\kappa tA_{w}(\left\Vert \xi\right\Vert _{p})}\widehat{u_{0}}\left(
\xi\right)  \in S\subset L^{2}$ for $t\geq0$, i.e. $u(x,t)\in
Dom(\boldsymbol{W})$ for $t\geq0$. Now
\begin{align*}
(\boldsymbol{W}u)(x,t)  &  =-\kappa\mathcal{F}_{\xi\rightarrow x}^{-1}\left(
A_{w}(\left\Vert \xi\right\Vert _{p})\mathcal{F}_{\xi\rightarrow x}\left(
u(x,t)\right)  \right) \\
&  =-\kappa\mathcal{F}_{\xi\rightarrow x}^{-1}\left(  A_{w}(\left\Vert
\xi\right\Vert _{p})e^{-tA_{w}(\left\Vert \xi\right\Vert _{p})}\widehat{u_{0}%
}\left(  \xi\right)  \right) \\
&  =-\kappa\mathcal{F}_{\xi\rightarrow x}^{-1}\left(  A_{w}(\left\Vert
\xi\right\Vert _{p})e^{-tA_{w}(\left\Vert \xi\right\Vert _{p})}1_{B_{R}}%
(\xi)\widehat{u_{0}}\left(  \xi\right)  \right) \\
&  =-\kappa\mathcal{F}_{\xi\rightarrow x}^{-1}\left(  e^{-\kappa
tA_{w}(\left\Vert \xi\right\Vert _{p})}A_{w}(\left\Vert \xi\right\Vert
_{p})1_{B_{R}}(\xi)\right)  \ast u_{0}(x).
\end{align*}

\end{proof}

As a direct consequence of Lemma \ref{lemma9} we obtain the following result.

\begin{proposition}
\label{prop2}Assume that $u_{0}\in S$. Then function $u\left(  x,t\right)  $
defined in (\ref{funct_u}) is a solution of Cauchy problem (\ref{Cauchy1}).
\end{proposition}

\subsection{Homogeneous equations with initial values in $L^{2}$}

We define
\begin{equation}
T(t)u=\left\{
\begin{array}
[c]{ll}%
Z_{t}\ast u, & t>0\\
& \\
u, & t=0,
\end{array}
\right.  \label{Oper_T}%
\end{equation}
for $u\in L^{2}$.

\begin{lemma}
\label{lemma10}The operator $T(t):L^{2}(\mathbb{Q}_{p}^{n})\longrightarrow
L^{2}(\mathbb{Q}_{p}^{n})$ is bounded for any fixed $t\geq0$.
\end{lemma}

\begin{proof}
For $t>0$, the result follows from the Young inequality by using the fact that
$Z_{t}\in L^{1}$, c.f. Theorem \ref{Thm1} (iii).
\end{proof}

\begin{proposition}
\label{prop3} The following assertions hold.

\noindent(i) The operator $\boldsymbol{W}$ generates a $C_{0}$ semigroup
$(\mathcal{T}(t))_{t\geq0}$. The operator $\mathcal{T}(t)$ coincides for each
$t\geq0$ with the operator $T(t)$ given by (\ref{Oper_T}).

\noindent(ii) Cauchy problem (\ref{Cauchy1}) is well-posed and its solution is
given by $u(x,t)=$ $Z_{t}\ast u_{0}$, $t\geq0$.
\end{proposition}

\begin{proof}
(i) \ By Proposition \ref{prop1} (iii) the operator $\boldsymbol{W}$ generates
a $C_{0}$ semigroup $(\mathcal{T}(t))_{t\geq0}$. Hence Cauchy problem
(\ref{Cauchy1}) is well-posed, see e.g. \cite[Theorem 3.1.1]{C-H}. By
Proposition \ref{prop2}, $\mathcal{T}(t)|_{S}=T(t)|_{S}$ and both operators
$\mathcal{T}(t)$ and $T(t)$ are defined on the whole $L^{2}$ and bounded, c.f.
Lemma \ref{lemma10}. By the continuity we conclude that $\mathcal{T}(t)|=T(t)$
on $L^{2}$. Now the statements follow from well-known results of the semigroup
theory, see e.g. \cite[Theorem 3.1.1]{C-H}, \cite[Chap. 2, Proposition
6.2]{E-N}.
\end{proof}

\subsection{Non-homogeneous equations}

Consider the following Cauchy problem:%

\begin{equation}
\left\{
\begin{array}
[c]{ll}%
\frac{\partial u}{\partial t}(x,t)-\boldsymbol{W}u(x,t)=g(x,t), &
x\in\mathbb{Q}_{p}^{n},t\in\left[  0,T\right]  ,T>0,\\
& \\
u\left(  x,0\right)  =u_{0}(x), & u_{0}(x)\in Dom(\boldsymbol{W}).
\end{array}
\right.  \label{Cauchy2}%
\end{equation}

We say that a function $u(x,t)$ is a solution of (\ref{Cauchy2}), if
$u(x,t)\in C\left(  [0,T),Dom(\boldsymbol{W})\right)  \cap C^{1}\left(
[0,T],L^{2}(\mathbb{Q}_{p}^{n})\right)  $ and if $u(x,t)$ satisfies equation
(\ref{Cauchy2}) for $t\in\lbrack0,T]$.

\begin{theorem}
\label{Thm3}Assume that $u_{0}\in Dom(\boldsymbol{W})$ and that $g\in C\left(
[0,\infty),L^{2}(\mathbb{Q}_{p}^{n})\right)  \cap L^{1}\left(  (0,\infty
),Dom(\boldsymbol{W})\right)  $. Then Cauchy problem (\ref{Cauchy2}) has a
unique solution given by%
\[
u(x,t)=\underset{\mathbb{Q}_{p}^{n}}{\int}Z(x-\xi,t)u_{0}(\xi)d^{n}\xi+%
{\displaystyle\int\limits_{0}^{t}}
\underset{\mathbb{Q}_{p}^{n}}{\int}Z(x-\xi,t-\theta)g(\xi,\theta)d^{n}\xi
d\theta.
\]

\end{theorem}

\begin{proof}
The result follows from Proposition \ref{prop3} \ by using some well-known
results of the semigroup theory, see e.g. \cite[Proposition 4.1.6]{C-H}.
\end{proof}

\begin{remark}
In \cite[Sect. 10.3]{A-K-S} a general theory for Cauchy Problems involving
pseudodifferential operators on Lizorkin spaces was developed. By applying
Proposition \ref{prop1} (i) and Theorem 10.3.1 in \cite{A-K-S} we can solve
Cauchy problem (\ref{Cauchy2}) when $u_{0}(x)$\ and $g(x,t)$ belong to a
certain Lizorkin space. However the results obtained by using this approach
are not sufficient for the purposes of application to stochastic processes, in
particular we need several properties of the semigroup associated with
operator $\boldsymbol{W}$.
\end{remark}

\section{\label{Sect6}First Passage Time Problem}

Consider the following Cauchy problem:%

\begin{equation}
\left\{
\begin{array}
[c]{ccc}%
\frac{\partial\varphi(x,t)}{\partial t}=\boldsymbol{W}\varphi(x,t), & t>0, &
x\in\mathbb{Q}_{p}^{n},\\
&  & \\
\varphi(x,t)=\Omega\left(  \left\Vert x\right\Vert _{p}\right)  . &  &
\end{array}
\right.  \label{Cauchy3}%
\end{equation}

In this section we assume that $\kappa$ satisfies%
\begin{equation}
\kappa{\int\limits_{\left\Vert y\right\Vert _{p}>1}}\frac{d^{n}y}{w\left(
\left\Vert y\right\Vert _{p}\right)  }\leq1. \label{Def_k}%
\end{equation}

By Theorem \ref{Thm3}\ the solution of (\ref{Cauchy3}) is $\varphi
(x,t)=Z_{t}\left(  x\right)  \ast\Omega\left(  \left\Vert x\right\Vert
_{p}\right)  $, and by Theorem \ref{Thm2} $Z_{t}$ is transition density of a
time and space homogeneous Markov process $X\left(  t,\omega\right)  $\ whose
paths has only first class discontinuities. Set $\Upsilon$ to be the space of
all paths of the random process $X\left(  t,\omega\right)  $. Then there
exists a probability space $\left(  \Upsilon,\mathcal{B},P\right)  $, where
$P$ is a probability measure on $\Upsilon$. The construction of this
probability space follows from classical arguments, see e.g. \cite[pp.
338-339]{Nelson} or \cite[proof of Theorem 5.9]{Lor-et-al}. The argument uses
the one-point compactification $\overline{\mathbb{Q}}_{p}^{n}$ of
$\mathbb{Q}_{p}^{n}$ by a point, and that $\Upsilon=%
{\textstyle\prod\limits_{0\leq t<\infty}}
\overline{\mathbb{Q}}_{p}^{n}(t)$, where the $\overline{\mathbb{Q}}_{p}%
^{n}(t)$ are copies of $\overline{\mathbb{Q}}_{p}^{n}$. The construction of
$P$ follows from the Stone-Weierstrass Theorem and Riesz-Markov Theorem like
in the Archimedean case. The probability $P\left(  d\omega\right)  $ is
roughly $%
{\textstyle\prod\limits_{i=1}^{+\infty}}
\left[  Z_{t_{i}}\left(  x_{i}\right)  \ast\Omega\left(  \left\Vert
x_{i}\right\Vert _{p}\right)  \right]  dx_{i}$.

This section is dedicated to the study of the following random variable.

\begin{definition}
The random variable $\tau_{\mathbb{Z}_{p}^{n}}:\,\Upsilon\rightarrow
\mathbb{R}_{+}\cup\left\{  +\infty\right\}  $ defined by
\[
\inf\{t>0;X(t,\omega)\in\mathbb{Z}_{p}^{n}\mid\text{there exists }t^{\prime
}\text{ such that }0<t^{\prime}<t\text{ and }X(t^{\prime},\omega
)\notin\mathbb{Z}_{p}^{n}\}
\]
is called the first passage time of a path of the random process $X(t,\omega)$
entering the domain $\mathbb{Z}_{p}^{n}$.
\end{definition}

Note that the initial condition in (\ref{Cauchy3}) implies that%
\[
P\left(  \left\{  \omega\in\Upsilon;X(0,\omega)\in\mathbb{Z}_{p}^{n}\right\}
\right)  =1.
\]

\begin{definition}
We say that $X(t,\omega)$ is recurrent \ with respect to $\mathbb{Z}_{p}^{n}$
if%
\begin{equation}
P\left(  \left\{  \omega\in\Upsilon;\tau_{\mathbb{Z}_{p}^{n}}(\omega
)\,<+\infty\right\}  \right)  =1. \label{Pr}%
\end{equation}
Otherwise we say that $X(t,\omega)$ is transient\ with respect to
$\mathbb{Z}_{p}^{n}$.
\end{definition}

The meaning of (\ref{Pr}) is that every path of $X(t,\omega)$ \ is sure to
return to $\mathbb{Z}_{p}^{n}$. If (\ref{Pr})\ does not hold, then there exist
paths of $X(t,\omega)$\ that abandon $\mathbb{Z}_{p}^{n}$ and never go back.

\begin{lemma}
\label{lemma11}The function $\varphi(x,t)=Z_{t}\left(  x\right)  \ast
\Omega\left(  \left\Vert x\right\Vert _{p}\right)  $ is infinitely
differentiable in the time $t\geq0$ and its derivative is given by%
\begin{equation}
\frac{\partial^{m}\varphi}{\partial t^{m}}(x,t)=\left(  -\kappa\right)
^{m}\int_{\mathbb{Q}_{p}^{n}}\left[  A(\left\Vert \xi\right\Vert _{p})\right]
^{m}\Psi(\xi\cdot x)\Omega(\left\Vert \xi\right\Vert _{p})e^{-\kappa
tA(\left\Vert \xi\right\Vert _{p})}d^{n}\xi\text{ for }m\in\mathbb{N}.
\label{eq:function u'(x,t)}%
\end{equation}

\end{lemma}

\begin{proof}
Note that for $t\geq0$ and $m\in\mathbb{N}$, \ $\left[  A(\left\Vert
\xi\right\Vert _{p})\right]  ^{m}\Psi(\xi\cdot x)\Omega(\left\Vert
\xi\right\Vert _{p})e^{-\kappa tA(\left\Vert \xi\right\Vert _{p})}\in
L^{1}(\mathbb{Q}_{p}^{n})$. The announced formula is obtained by induction on
$m$ by applying the Lebesgue Dominated Convergence Theorem.
\end{proof}

\begin{lemma}
\label{lemma12}The probability density function for a path of $X(t,\omega)$ to
enter into $\mathbb{Z}_{p}^{n}$ at the instant of time $t$, with the condition
that $X(0,\omega)\in$ $\mathbb{Z}_{p}^{n}$ is given by%
\begin{equation}
g(t)=\kappa\underset{\mathbb{Q}_{p}^{n}\smallsetminus\mathbb{Z}_{p}^{n}}{\int
}\frac{\varphi(x,t)}{w(\left\Vert x\right\Vert _{p})}d^{n}x. \label{g(t)}%
\end{equation}

\end{lemma}

\begin{proof}
The survival probability, by definition%
\[
S(t):=S_{\mathbb{Z}_{p}^{n}}(t)=\underset{\mathbb{Z}_{p}^{n}}{\int}%
\varphi(x,t)d^{n}x,
\]
is the probability \ that a path of $X(t,\omega)$ remains in $\mathbb{Z}%
_{p}^{n}$\ at the time $t$. Because there are no external forces acting on the
random walk, we have
\begin{align}
S^{\prime}(t)  &  =%
\begin{array}
[c]{l}%
\text{Probability that a path of }X(t,\omega)\\
\text{goes back to }\mathbb{Z}_{p}^{n}\text{ at the time }t
\end{array}
-%
\begin{array}
[c]{l}%
\text{Probability that a path of }X(t,\omega)\\
\text{exits }\mathbb{Z}_{p}^{n}\text{ at the time }t
\end{array}
\label{S}\\
&  =g(t)-C\cdot S(t)\text{ with }0<C\leq1.\nonumber
\end{align}

By using Lemma \ref{lemma11} and the definition of $\boldsymbol{W}$, we have%

\begin{align*}
S^{\prime}(t)  &  =\underset{\mathbb{Z}_{p}^{n}}{\int}\frac{\partial
\varphi(x,t)}{\partial t}d^{n}x=\kappa\underset{\mathbb{Z}_{p}^{n}}{\int\text{
}}\underset{\mathbb{Q}_{p}^{n}}{\int}\frac{\varphi(x+y,t)-\varphi
(x,t)}{w(\left\Vert y\right\Vert _{p})}d^{n}yd^{n}x\\
&  =\underset{\mathbb{Z}_{p}^{n}}{\kappa\int}\underset{\mathbb{Z}_{p}^{n}%
}{\text{ }\int}\frac{\varphi(x+y,t)-\varphi(x,t)}{w(\left\Vert y\right\Vert
_{p})}d^{n}yd^{n}x+\kappa\underset{\mathbb{Z}_{p}^{n}}{\int}\text{ }%
\underset{\mathbb{Q}_{p}^{n}\smallsetminus\mathbb{Z}_{p}^{n}}{\int}%
\frac{\varphi(x+y,t)-\varphi(x,t)}{w(\left\Vert y\right\Vert _{p})}d^{n}%
yd^{n}x\\
&  =\kappa\underset{\mathbb{Z}_{p}^{n}}{\int}\text{ }\underset{\mathbb{Q}%
_{p}^{n}\smallsetminus\mathbb{Z}_{p}^{n}}{\int}\frac{\varphi(x+y,t)}%
{w(\left\Vert y\right\Vert _{p})}d^{n}yd^{n}x-\kappa\underset{\mathbb{Z}%
_{p}^{n}}{\int}\text{ }\underset{\mathbb{Q}_{p}^{n}\smallsetminus
\mathbb{Z}_{p}^{n}}{\int}\frac{\varphi(x,t)}{w(\left\Vert y\right\Vert _{p}%
)}d^{n}yd^{n}x\\
&  =\kappa\underset{\mathbb{Z}_{p}^{n}}{\int}\text{ }\underset{\mathbb{Q}%
_{p}^{n}\smallsetminus\mathbb{Z}_{p}^{n}}{\int}\frac{\varphi(z,t)}%
{w(\left\Vert z\right\Vert _{p})}d^{n}zd^{n}x-\underset{\mathbb{Z}_{p}^{n}%
}{\int}\varphi(x,t)d^{n}x\underset{\mathbb{Q}_{p}^{n}\smallsetminus
\mathbb{Z}_{p}^{n}}{\int}\frac{\kappa}{w(\left\Vert y\right\Vert _{p})}%
d^{n}y\\
&  =\underset{\mathbb{Q}_{p}^{n}\smallsetminus\mathbb{Z}_{p}^{n}}{\kappa\int
}\frac{\varphi(z,t)}{w(\left\Vert z\right\Vert _{p})}d^{n}z-\left(
\underset{\mathbb{Q}_{p}^{n}\smallsetminus\mathbb{Z}_{p}^{n}}{\int}%
\frac{\kappa}{w(\left\Vert y\right\Vert _{p})}d^{n}y\right)  S(t).
\end{align*}
Take $C=\underset{\mathbb{Q}_{p}^{n}\smallsetminus\mathbb{Z}_{p}^{n}}{\int
}\frac{\kappa}{w(\left\Vert y\right\Vert _{p})}d^{n}y\leq1$, c.f.
(\ref{Def_k}). Finally, by using (\ref{S}), one gets
\[
g(t)=\kappa\underset{\mathbb{Q}_{p}^{n}\smallsetminus\mathbb{Z}_{p}^{n}}{\int
}\frac{\varphi(x,t)}{w(\left\Vert x\right\Vert _{p})}d^{n}x.
\]

\end{proof}

\begin{proposition}
\label{prop4}The probability density function $f(t)$ of the random variable
$\tau(\omega)$ satisfies the non-homogeneous Volterra equation of second kind%
\begin{equation}
g(t)=\int_{0}^{\infty}g(t-\tau)f(\tau)d\tau+f(t). \label{eq:Volterra}%
\end{equation}

\end{proposition}

\begin{proof}
The result follows from Lemma \ref{lemma12}\ by using the argument given in
the proof \ of Theorem 1 in \cite{Av-2}.
\end{proof}

\begin{proposition}
\label{prop5}The Laplace transform $G(s)$ of $g(t)$ is given by%
\begin{equation}
G(s)=\kappa^{2}(1-p^{-n})\sum_{i=1}^{\infty}\frac{p^{in}}{w(p^{i})}\sum
_{j=i}^{\infty}\frac{\frac{p^{n}-2}{w(p^{j+1})}+\frac{p^{-n}}{w(p^{j})}%
}{\left(  s+\kappa A_{w}(p^{-j})\right)  \left(  s+\kappa A_{w}(p^{-j+1}%
)\right)  }\quad\text{for}\quad\operatorname{Re}(s)>0. \label{G_s}%
\end{equation}

\end{proposition}

\begin{proof}
We first note that%
\begin{equation}
\frac{e^{-st}e^{-\kappa tA_{w}\left(  \left\Vert \xi\right\Vert _{p}\right)
}\Omega\left(  \left\Vert \xi\right\Vert _{p}\right)  }{w\left(  \left\Vert
x\right\Vert _{p}\right)  }\in L^{1}\left(  \left(  0,\infty\right)
\times\mathbb{Q}_{p}^{n}\times\mathbb{Q}_{p}^{n}\smallsetminus\mathbb{Z}%
_{p}^{n},dtd^{n}\xi d^{n}x\right)  \text{ for }\operatorname{Re}\left(
s\right)  >0. \label{condition_2}%
\end{equation}
We compute the Laplace transform $G(s)$ of $g(t)$ by\ replacing
\[
\varphi(x,t)=\underset{\mathbb{Q}_{p}^{n}}{\int}e^{-\kappa tA_{w}(\left\Vert
\xi\right\Vert _{p})}\Omega(\left\Vert \xi\right\Vert _{p})\Psi(\xi\cdot
x)d^{n}\xi
\]
in (\ref{g(t)}) and interchanging the iterated integrals in a suitable form,
which is allowed by (\ref{condition_2}) via Fubini's Theorem, in this way one
gets%
\[
G(s)=\kappa\underset{\mathbb{Q}_{p}^{n}\smallsetminus\mathbb{Z}_{p}^{n}}{\int
}\text{ }\underset{\mathbb{Z}_{p}^{n}}{\int}\frac{\Psi(\xi\cdot x)}{\left(
s+\kappa A_{w}(\left\Vert \xi\right\Vert _{p})\right)  w(\left\Vert
x\right\Vert _{p})}d^{n}\xi d^{n}x\text{ for }\operatorname{Re}\left(
s\right)  >0.
\]

We now assert that%
\[
\frac{1}{\left\vert s+\kappa A_{w}(\left\Vert \xi\right\Vert _{p})\right\vert
w(\left\Vert x\right\Vert _{p})}\in L^{1}\left(  \mathbb{Q}_{p}^{n}%
\smallsetminus\mathbb{Z}_{p}^{n}\times\mathbb{Z}_{p}^{n},d^{n}xd^{n}%
\xi\right)  \text{ for }\operatorname{Re}\left(  s\right)  >0.
\]
Indeed, since%
\begin{equation}
\left\vert s+\kappa A_{w}(\left\Vert \xi\right\Vert _{p})\right\vert
\geq\operatorname{Re}\left(  s\right)  +\kappa A_{w}(\left\Vert \xi\right\Vert
_{p})>\kappa A_{w}(\left\Vert \xi\right\Vert _{p})\text{ for }%
\operatorname{Re}\left(  s\right)  >0, \label{linea}%
\end{equation}
we have
\begin{align*}
\frac{1}{\left\vert s+\kappa A_{w}(\left\Vert \xi\right\Vert _{p})\right\vert
w(\left\Vert x\right\Vert _{p})}  &  <\frac{1}{\kappa A_{w}(\left\Vert
\xi\right\Vert _{p})w(\left\Vert x\right\Vert _{p})}\\
&  \leq\frac{C}{\left\Vert \xi\right\Vert _{p}^{\alpha_{2}-n}w(\left\Vert
x\right\Vert _{p})}\leq\frac{C^{\prime}}{\left\Vert \xi\right\Vert
_{p}^{\alpha_{2}-n}\left\Vert x\right\Vert _{p}^{\alpha_{1}}},
\end{align*}
which is an integrable function for $x\in\mathbb{Q}_{p}^{n}\smallsetminus
\mathbb{Z}_{p}^{n}$, $\xi\in\mathbb{Z}_{p}^{n}$ and $\operatorname{Re}\left(
s\right)  >0$, c.f. Lemma \ref{lemma3}.

In order to calculate an explicit formula for $G(s)$ for $\operatorname{Re}%
\left(  s\right)  >0$ we proceed as follows. We take $U=\left\{
y\in\mathbb{Q}_{p}^{n};\left\Vert y\right\Vert _{p}=1\right\}  $ as before,
then $\mathbb{Q}_{p}^{n}\smallsetminus\mathbb{Z}_{p}^{n}=%
{\textstyle\bigsqcup\nolimits_{i\in\mathbb{N}\smallsetminus\left\{  0\right\}
}}
p^{-i}U$, $\mathbb{Z}_{p}^{n}\smallsetminus\left\{  0\right\}  =%
{\textstyle\bigsqcup\nolimits_{j\in\mathbb{N}}}
p^{j}U$, and
\[
G(s)=\kappa%
{\displaystyle\sum\limits_{i=1}^{\infty}}
{\displaystyle\sum\limits_{j=0}^{\infty}}
\underset{p^{-i}U}{\int}\text{ }\underset{p^{j}U}{\int}\frac{\Psi(\xi\cdot
x)}{\left[  s+\kappa A_{w}(\left\Vert \xi\right\Vert _{p})\right]
w(\left\Vert x\right\Vert _{p})}d^{n}\xi d^{n}x\text{ for }\operatorname{Re}%
\left(  s\right)  >0\text{.}%
\]

We now use the following change of variables:%
\[%
\begin{array}
[c]{ccc}%
p^{-i}U\times p^{j}U & \rightarrow & U\times U\\
&  & \\
\left(  x,\xi\right)  & \rightarrow & \left(  y^{\prime},y\right)
\end{array}
\]
with $x=p^{-i}y^{\prime}$ for $i\in\mathbb{N\smallsetminus}\left\{  0\right\}
$, $\xi=p^{j}y$ for $j\in\mathbb{N}$. Furthermore, $d^{n}\xi d^{n}%
x=p^{-nj+ni}d^{n}yd^{n}y^{\prime}$ for $j\in\mathbb{N}$ and $i\in
\mathbb{N\smallsetminus}\left\{  0\right\}  $. Then
\begin{align*}
G(s)  &  =\kappa\sum_{i=1}^{\infty}p^{in}\underset{U}{\int}\sum_{j=0}^{\infty
}p^{-jn}\underset{U}{\int}\frac{\Psi(p^{j-i}y\cdot y^{\prime})}{\left[
s+\kappa A_{w}(p^{-j})\right]  w(p^{i})}d^{n}yd^{n}y^{\prime}\\
&  =\kappa\sum_{i=1}^{\infty}\frac{p^{in}}{w(p^{i})}\sum_{j=0}^{\infty}%
\frac{p^{-jn}}{s+\kappa A_{w}(p^{-j})}\underset{U}{\int}\underset{U}{\int}%
\Psi(p^{j-i}y\cdot y^{\prime})d^{n}yd^{n}y^{\prime}\text{ (see (\ref{formula}%
))}\\
&  =\kappa(1-p^{-n})\sum_{i=1}^{\infty}\frac{p^{in}}{w(p^{i})}\left(
\sum_{j=i}^{\infty}\frac{(1-p^{-n})p^{-jn}}{s+\kappa A_{w}(p^{-j})}%
-\frac{p^{-n}p^{-n(i-1)}}{s+\kappa A_{w}(p^{1-i})}\right) \\
&  =\kappa(1-p^{-n})\sum_{i=1}^{\infty}\frac{p^{in}}{w(p^{i})}\sum
_{j=i}^{\infty}p^{-jn}\left(  \frac{1}{s+\kappa A_{w}(p^{-j})}-\frac
{1}{s+\kappa A_{w}(p^{-j+1})}\right)  \text{ (see (\ref{fomula_A}))}\\
&  =\kappa^{2}(1-p^{-n})\sum_{i=1}^{\infty}\frac{p^{in}}{w(p^{i})}\sum
_{j=i}^{\infty}\frac{\frac{p^{n}-2}{w(p^{j+1})}+\frac{p^{-n}}{w(p^{j})}%
}{\left(  s+\kappa A_{w}(p^{-j})\right)  \left(  s+\kappa A_{w}(p^{-j+1}%
)\right)  }.
\end{align*}

\end{proof}

\begin{theorem}
\label{Thm4}(i) If $\boldsymbol{W}$ is of polynomial type ($\alpha_{3}=0$,
$\alpha_{1}=\alpha_{2}=\alpha>n$) and $\alpha\geq2n$, then $X\left(
t,\omega;\boldsymbol{W}\right)  $ is recurrent with respect to $\mathbb{Z}%
_{p}^{n}$.

\noindent(ii) If $\boldsymbol{W}$ is of polynomial type ($\alpha_{3}=0$,
$\alpha_{1}=\alpha_{2}=\alpha>n$) and $n<\alpha<2n$, then $X\left(
t,\omega;\boldsymbol{W}\right)  $ is transient with respect to $\mathbb{Z}%
_{p}^{n}$.
\end{theorem}

\begin{proof}
By Proposition \ref{prop4}, the Laplace transform of $F(s)$ of $f(t)$ equals
$\frac{G(s)}{1+G(s)}$, where $G(s)$ is the Laplace transform of $g(t)$, and
thus $F\left(  0\right)  =\int_{0}^{\infty}f\left(  t\right)  dt=1-\frac
{1}{1+G(0)}$. Hence in order to prove that $X\left(  t,\omega;\boldsymbol{W}%
\right)  $ is recurrent is sufficient to show that $G(0)=\lim_{s\rightarrow
0}G(s)=\infty$, and to prove that it is transient that $G(0)=\lim
_{s\rightarrow0}G(s)<\infty$.

(i) Take $s\in\mathbb{R}$, $s>0$ and $\alpha_{3}=0$, $\alpha_{1}=\alpha
_{2}=\alpha>n$. First we note that%
\begin{equation}
C_{2}p^{-j\left(  \alpha-n\right)  }\leq A_{w}(p^{-j})\leq A_{w}(p^{-\left(
j-1\right)  })\leq C_{3}p^{-\left(  j-1\right)  \left(  \alpha-n\right)  },
\label{Ec10}%
\end{equation}

c.f. Lemma \ref{lemma3}. We take $s\in\mathbb{R}$ with $s>0$ and set%
\begin{equation}
s=C_{3}p^{-\left(  j_{0}(s)-1\right)  \left(  \alpha-n\right)  }. \label{Ec11}%
\end{equation}
Note that $s\rightarrow0^{+}\Leftrightarrow j_{0}:=j_{0}(s)\rightarrow\infty$.
Then%
\[
\left(  s+\kappa A_{w}(p^{-j})\right)  \left(  s+\kappa A_{w}(p^{-j+1}%
)\right)  \leq\left(  1+\kappa\right)  ^{2}s^{2}\text{ for }j\geq j_{0}%
\]

because $A_{w}(p^{-\gamma})$\ is a decreasing function of $\gamma$.

By (\ref{G_s}),
\begin{align}
G(s)  &  >\kappa^{2}(1-p^{-n})\frac{p^{n}}{w(p)}\sum_{j=1}^{\infty}\frac
{\frac{p^{n}-2}{w(p^{j+1})}+\frac{p^{-n}}{w(p^{j})}}{\left(  s+\kappa
A_{w}(p^{-j})\right)  \left(  s+\kappa A_{w}(p^{-j+1})\right)  }\label{Ec12}\\
&  >\kappa^{2}(1-p^{-n})\frac{p^{n}}{w(p)}\sum_{j=j_{0}}^{\infty}\frac
{\frac{p^{n}-2}{w(p^{j+1})}+\frac{p^{-n}}{w(p^{j})}}{\left(  s+\kappa
A_{w}(p^{-j})\right)  \left(  s+\kappa A_{w}(p^{-j+1})\right)  }\nonumber
\end{align}

with $j_{0}\in\mathbb{N}$. Now by (\ref{Ec5}) with $\alpha_{3}=0$, $\alpha
_{1}=\alpha_{2}=\alpha$,%
\[
\frac{p^{n}-2}{w(p^{j+1})}+\frac{p^{-n}}{w(p^{j})}\geq\left(  \frac{p^{n}%
-2}{C_{1}p}+\frac{p^{-n}}{C_{1}}\right)  p^{-j\alpha}\text{ for }%
j\in\mathbb{N}.
\]
By (\ref{Ec10})-(\ref{Ec11}) and (\ref{Ec12}), we have
\begin{align*}
G(s)  &  >\frac{C}{s^{2}}p^{-j_{0}\alpha}=C^{\prime}\frac{p^{-j_{0}\alpha}%
}{p^{-2j_{0}\left(  \alpha-n\right)  }}\\
&  =C^{\prime}p^{j_{0}\left(  \alpha-2n\right)  }.
\end{align*}

Hence, if $\alpha>2n$,%
\[
\lim_{s\rightarrow0+}G(s)>C^{\prime}\lim_{s\rightarrow0+}p^{j_{0}\left(
\alpha-2n\right)  }=C^{\prime}\lim_{j_{0}\rightarrow\infty}p^{j_{0}\left(
\alpha-2n\right)  }=\infty.
\]

Now if $\alpha=2n$ and $s\in\mathbb{R}$, $s>0$, we have%

\[
\lim_{j\rightarrow\infty}\frac{\frac{p^{n}-2}{w(p^{j+1})}+\frac{p^{-n}%
}{w(p^{j})}}{\left(  s+\kappa A_{w}(p^{-j})\right)  \left(  s+\kappa
A_{w}(p^{-j+1})\right)  }\geq C\lim_{j\rightarrow\infty}p^{j\left(
\alpha-2n\right)  }=C
\]
where $C$ is a positive constant. Hence $\lim_{s\rightarrow0+}G(s)=\infty$.

(ii) In this case $\alpha_{3}=0$, $\alpha_{1}=\alpha_{2}=\alpha>n$, and
$s\in\mathbb{C}$ with $\operatorname{Re}(s)>0$. By (\ref{Ec5}),%
\[
\frac{p^{n}-2}{w(p^{j+1})}+\frac{p^{-n}}{w(p^{j})}\leq\left(  \frac{p^{n}%
-2}{C_{0}p}+\frac{p^{-n}}{C_{0}}\right)  p^{-j\alpha}%
\]
for $j\in\mathbb{N}$. In addition, by Lemma \ref{lemma3},
\[
\frac{1}{\left\vert \left(  s+\kappa A_{w}(p^{-j})\right)  \left(  s+\kappa
A_{w}(p^{-j+1})\right)  \right\vert }\leq\frac{1}{\left[  \kappa A_{w}%
(p^{-j})\right]  ^{2}}\leq\frac{1}{\kappa^{2}C_{2}^{2}p^{-2j\left(
\alpha-n\right)  }}%
\]
for $j\in\mathbb{N}$, c.f. (\ref{linea}).

Hence, by (\ref{G_s}),
\begin{align*}
G(s)  &  \leq C\sum_{i=1}^{\infty}p^{-i\left(  \alpha-n\right)  }\sum
_{j=i}^{\infty}p^{-j\alpha+2j\left(  \alpha-n\right)  }=C\sum_{i=1}^{\infty
}p^{-i\left(  \alpha-n\right)  }\sum_{j=i}^{\infty}p^{-j\left(  2n-\alpha
\right)  }\\
&  =C^{\prime}\sum_{i=1}^{\infty}p^{-in}<\infty\text{ if }2n>\alpha.
\end{align*}

\end{proof}

\end{document}